\documentclass[11pt,a4paper]{article}
\usepackage[top=70pt,bottom=70pt,left=60pt,right=60pt]{geometry}
\usepackage{bbold}
\usepackage{caption}
\usepackage{longtable}
\usepackage{cellspace}
\usepackage{booktabs}
\usepackage{arydshln}
\usepackage{amsfonts}
\usepackage{amssymb}
\usepackage{graphicx}
\usepackage{amsmath}
\usepackage{amsthm}
\usepackage{amsmath}
\usepackage{amsfonts}
\usepackage{amsthm}
\usepackage{mathabx}
\usepackage[colorlinks=false,linkcolor=blue,citecolor= red,bookmarksopen=true]{hyperref}
\usepackage{comment}
\usepackage{enumerate}
\usepackage{dsfont}
\usepackage{xcolor}
\usepackage{soul} 

\newcommand{\PWq}{{\mathbb Q}}

\newcommand{\R}{\mathbb{R}} 

\newcommand{\QW}{\mathbb{Q}}

\newcommand{\E}{\mathbb{E}}

\newcommand{\PW}{\mathbb{P}}

\newcommand{\Borel}{\mathbb{B}}

\newcommand{\Om}{\Omega}

\newcommand{\om}{\omega}

\newcommand{\F}{\mathbb{F}}

\newcommand{\Ep}[1]{\E_\PW\left[#1\right]}
\newcommand{\Eq}[1]{\E_\QW\left[#1\right]}

\newcommand{\Ezero}[1]{\E_0\left[#1\right]}
\newcommand{\Acal}{{\mathcal A}}

\newcommand{\Fcal}{{\mathcal F}}

\newcommand{\mbu}{{\mathbb u}}

\newtheorem{theorem}{Theorem}[section]
\newtheorem{proposition}[theorem]{Proposition}

\newtheorem{lemma}[theorem]{Lemma}

\newtheorem*{assumption}{Assumption}
\newtheorem{remark}[theorem]{Remark}
\newtheorem{example}[theorem]{Example}
\newtheorem{definition}[theorem]{Definition}

\theoremstyle{definition}

\renewcommand{\om}{\omega}
\renewcommand{\Om}{\Omega}
\renewcommand{\PW}{{\mathbb P}}
\renewcommand{\PWq}{{\mathbb Q}}
\renewcommand{\QW}{{\mathbb Q}}
\renewcommand{\R}{\mathbb{R}}
\renewcommand{\E}{\mathbb{E}}
\renewcommand{\Borel}{\mathbb{B}}

\renewcommand{\Ep}[1]{\E_\PW\left[#1\right]}
\renewcommand{\Eq}[1]{\E_\QW\left[#1\right]}

\renewcommand{\Ezero}[1]{\E_0\left[#1\right]}

\begin{document}
	
	
\title{On consistency of optimal portfolio choice for state-dependent exponential utilities\footnote{This is a revised version with respect to the published one}}
	
	\author{Edoardo Berton\thanks{Universit\`a Cattolica del Sacro Cuore, Milano, edoardo.berton@unicatt.it } \and Marzia De Donno\thanks{Universit\`a Cattolica del Sacro Cuore, Milano, marzia.dedonno@unicatt.it} \and Marco Maggis\thanks{Universit\`a degli Studi di Milano, marco.maggis@unimi.it \newline The authors thank Marta Minotti, Stefano Pagliarani and Alessandro Sbuelz for helpful discussions on this subject.}}
	
	\date{}
	
	\maketitle

\begin{abstract}
In a complete Brownian market, we characterize consistency across investment horizons for state-dependent exponential utilities. We show that consistency is equivalent to an explicit relation between the drift and diffusion of risk tolerance and a deterministic absolutely continuous function.  Within this family, the forward performance property selects a unique coefficient. More generally, under the stated admissibility conditions, a forward performance is consistent, whereas consistency alone does not imply the forward performance property.
\end{abstract}

\noindent \textbf{Keywords}:time consistency, state-dependent utility, portfolio choice, forward performance, exponential utility, market price of risk

\section{Introduction}

The problem of optimal portfolio choice stands out as one of the most significant streams of literature related to decision making under uncertainty (see \cite{Gilboa09} for an introduction to the theory of decisions). In this framework, an agent aims at maximizing her utility by investing in a financial market. Numerous approaches to optimal investment have been proposed since the seminal contributions of \cite{Samuelson69} and \cite{Merton69}, which extended this theory beyond the \cite{Markowitz53} one-period model.

In a multi-period setting, an optimal strategy is an investment plan that ensures the agent to maximize a function of her terminal wealth resulting from trading in the market. Ideally, such a strategy is determined at the beginning of the trading period and is intended to be followed `\textit{consistently}' until the specified terminal date. However, the optimization problem may be \textit{time-inconsistent}, meaning that the strategy currently considered optimal may not remain optimal in subsequent periods. Consequently, it is not universally clear what optimality should signify in such contexts.

The literature has addressed time inconsistency through two primary approaches. The first, known as the \textit{precommitment} approach, resolves inter-temporal inconsistency by committing to a strategy evaluated as optimal at inception, regardless of future deviations. This approach has been extensively studied in works such as \cite{LN00,ZL00,JZ08}. The second approach adopts a game-theoretic perspective, treating future incentives to deviate as constraints on optimization. This perspective, pioneered by \cite{EL06} and further developed by \cite{BC10}, \cite{BMZ14}, and \cite{BKM17}, often addresses time inconsistency arising from mean-variance utility formulations.

Classical studies, such as the Merton problem, focus on maximizing the expected utility $\mathbb{E}[\mbu(V_t)]$ at a terminal time $t$ by deriving optimal strategies through backward induction. In contrast, more recent research \cite{MZ07, MZ08, MZ09, HH07} introduces an alternative, maturity-independent framework. In this setting, the utility function evolves stochastically and forwardly over time. Known as forward performance theory, this framework substitutes the classical indirect utility function with a forward utility criterion derived from the financial market. The stochastic utility is deduced in order to satisfy specific (super)martingale conditions to recover the Dynamic Programming Principle.

In this paper, we approach portfolio choice in a simple market model from a `forward' perspective, focusing primarily on the issue of consistency. Unlike forward performance theory, our approach centers on deriving optimal strategies that can be consistently extended beyond a fixed terminal date in a self-financing manner.
In Definition \ref{def:time_consistency} below, we adapt and refine the seminal notion of consistency provided in \cite[p. 475]{KP77}. A policy $\alpha^\ast$ is consistent if, `for any time period $t$, it optimizes the objective functional, taking as given previous decisions and that future policy decisions are similarly selected'.

In Proposition \ref{prop:det_exp_ut}, we demonstrate that inconsistency may arise in the classical Merton problem when the market price of risk exhibits sufficient stochasticity. The rationale is that constant risk aversion prevents the agent from effectively counterbalancing randomness in the investment set, forcing her to later deviate from the initially devised strategy (see Remark \ref{failure} for additional details). However, this drawback is mitigated when uncertainty about the agent's future risk aversion is introduced. 
Our main result, stated in Theorem \ref{main:thm}, establishes that consistency for exponential-type utilities is achieved if and only if the agent's risk aversion evolves according to a dynamic determined by an explicit relation between the drift and diffusion.

The relation with forward performance theory requires a separate argument. A forward performance provides an optimizer at every horizon and therefore implies consistency whenever its optimal wealth belongs to the same admissible class, whereas the converse is false. Similarly to our paper, \cite{Zitkovic09} focuses on random fields with an exponential structure and provides necessary and sufficient conditions for the self-generation of preferences over time. In Proposition \ref{prop:infinite_new} identifies the forward performance within the normalized exponential fields considered here, matching in this way Theorem 4.4 in \cite{Zitkovic09}. Remark  \ref{rem:implications} also shows why consistency alone is insufficient to establish the forward performance property.

\section{Preliminaries and problem formulation}
The probability space $(\Omega ,\Fcal,\mathbb{P})$ is fixed throughout
the paper and we shall denote by $L^{0}(\Omega ,\mathcal{F},\mathbb{P})$
the space of $\mathcal{F}$ measurable random variables that are $\mathbb{P}$
a.s. finite. Similarly $L^{1}(\Omega ,\mathcal{F},\mathbb{P})$ will denote the space of integrable random variables. In this paper a filtration $\mathbb{F} := (\Fcal_t)_{t\geq 0}$ will describe the information available to an economic agent at any time $t\in [0,\infty)$ and we shall always assume that the filtered probability space $(\Omega,\Fcal,\F,\PW)$ satisfies the usual conditions, that is the filtration $\F$ is right-continuous and complete. 
\\ A decision maker aims at maximizing her utility by investing her initial endowment $x\in \R$ in a market 
composed by a risky asset $(S_t)_{t\geq 0}$ and a bond $(B_t)_{t\geq 0}$ whose dynamics are as follows 
\begin{equation*}
dS_t = S_t\left(\mu_t dt+ \sigma_t dW_t\right), \quad dB_t=r B_t dt
\end{equation*}
with $S_0>0$, $B_0=1$, $W$ being a $1$-dimensional Brownian motion and $(\mu_t)_{t\ge 0},(\sigma_t)_{t\ge 0}$ $\F$-progressively measurable processes. We assume that the filtration $\F$ is the completion of the natural filtration generated by $(W_t)_{t\geq0}$, therefore the market is complete. Moreover, for the sake of the exposition of this decision problem, we also assume that 
\begin{equation}\label{hp_theta} \PW\left(\theta_t\neq 0\text{ and } \sup_{0\leq u\leq t}|\theta_u|<K_t\right)=1 \text{ for any } t>0 \text{ and some } K_t\in [0,\infty).\end{equation}
As a result, the process 
$$Z_t = \exp\left\{-\frac12\int_0^t \theta^2_s ds + \int_0^t \theta_s dW_s\right\} \text{ where } \theta_t = -\left(\frac{\mu_t - r}{\sigma_t}\right)$$
is a $(\F, \PW)$-martingale and for any $0< t< \infty$ the probability measure $\PWq_t$, with $Z_t = \frac{d\PWq_t}{d\PW}$, defines a martingale measure for the discounted price process $\left(e^{-ru}S_u\right)_{u\in[0,t]}$. We observe that, for any choice of $t,s$ such that $t>s$, the restriction of $\PWq_t$ to $\Fcal_s$ coincides with $\PWq_s$. Hence, to avoid burdensome notation, we omit the subscript and denote by $\PWq$ the family of measures\footnote{We remark that this is only a notational convention and does not imply the existence of a probability measure on $\Fcal$, unless additional technical conditions are imposed (see Remark \ref{rem:implications}).} induced by the process $(Z_t)_{t\geq 0}$. In particular we adopt the following convention, for an $\Fcal_t$-measurable random variable $X$ we define $\Eq{X} := \E_{\PWq_t}\left[X\right] = \Ep{Z_t X}$ as well as $\Eq{X \middle| \Fcal_s} := \E_{\PWq_t}\left[X\middle | \Fcal_s \right] = \frac{\Ep{Z_t X \middle|\Fcal_s}}{Z_s}$, for $s<t$. Additionally, we shall refer to the process $W^\PWq$ as $\PWq$-Brownian motion, with the interpretation that for any $t>0$ the process $(W_s^\PWq)_{s\in[0, t]}$ is a Brownian motion under the measure $\PWq_t$.

Throughout this paper, all the financial quantities of interest are to be considered in discounted terms. Hence, $(V^\alpha_t)_{t\geq0}$ denotes the discounted value of a self-financing portfolio and $\alpha_t$ 
represents the proportion of $V^\alpha_t$ invested in the risky asset at time $t$. The dynamics of the portfolio process are given by 
\begin{equation}\label{discounted:port}
    dV^\alpha_t = V^\alpha_t \left[(\mu_t - r) \alpha_t dt + \sigma_t \alpha_t dW_t\right]\quad V_0^{\alpha}=v_0.
\end{equation}

\medskip

For a fixed time horizon $t\in (0,+\infty)$, the agent is endowed with a preference relation $\succeq^t$ that admits a numerical representation through a state-dependent utility $\mbu_t:\Omega\times \R\to \R\cup \{-\infty\}$ and the subjective probability of the agent coincides with the reference\footnote{This can be assumed without loss of generality if the subjective probability is equivalent to the reference measure, as the function $\mbu_t$ is unique up to a change of measure $\PW$ (see \cite{WZ99})}. With an abuse of notation we shall always denote by $\mbu_t(X)$ the random variable $\omega\mapsto \mbu_t(\omega,X(\omega))$ for any $X\in L^{0}(\Omega,\Fcal_t,\PW)$. Given $X,Y\in L^{0}(\Omega,\Fcal_t,\PW)$ such that $\mbu_t(X),\mbu_t(Y)\in L^{1}(\Omega,\Fcal_t,\PW)$, the preference ordering is given by 
\begin{equation}\label{repr:T}
 X \succeq^t Y    \quad \text{ if and only if } \quad  \Ep{\mbu_t(X)}\geq \Ep{\mbu_t(Y)}.
\end{equation}
\cite{WZ99} proved that this representation holds if the preference $\succeq^t$ is monotone, pointwise continuous, and satisfies the Sure Thing principle. Furthermore, \cite{BDM24} showed in addition that $\mbu_t$ can be chosen to be $\Fcal_t\otimes \mathbb{B}_{\R}$ measurable\footnote{$\mathbb{B}_{\R}$ is the Borel $\sigma$-algebra on the real line, the $\sigma$-algebra generated by open sets of $\R$}. with $\mbu_t(\omega,\cdot)$ strictly increasing and continuous for every $\omega\in\Omega$. 
\\ In the decision framework just depicted, the classical utility maximization problem thus takes the form 
\begin{equation}\label{optimization}\text{maximize }  \Ep{\mbu_t(V^\alpha_t)} \text{ constrained to } V^\alpha_0=v_0,
\end{equation}
where $\alpha$ belongs to a suitable family $\mathcal{A}$ of admissible controls and $x \in \R$ represents the initial wealth of the agent. In particular, we will always suppose that any admissible control $\alpha\in\mathcal{A}$ guarantees the existence and uniqueness of a (weak) solution of \eqref{discounted:port}. We shall denote by $(\alpha^\ast_u)_{u\in [0,t]}$ the strategy that maximizes the expected utility of the agent in the time interval $[0,t]$ and, therefore, $V_t^{\alpha^\ast}$ denotes the maximizer of the objective functional $ \Ep{\mbu_t(\cdot)}$. 

\medskip

As illustrated in the introduction, we shall adopt throughout the paper the following notion of consistency, which rephrases the seminal notion in \cite{KP77}. 

\begin{definition}\label{def:time_consistency} For $s < t$, let $V_s^{\alpha^\ast}$ be the maximizer of $\Ep{\mbu_s(\cdot)}$ and $V_t^{\alpha^\ast}$ be the maximizer of $\Ep{\mbu_t(\cdot)}$ both selected by solving \eqref{optimization}. An optimization problem is \textbf{consistent} with respect to any time horizon if for any $s<t$ it holds that\footnote{We stress that this condition is equivalent to $\Ep{Z_t V_t^{\alpha^\ast}\middle| \Fcal_s} = Z_s V_s^{\alpha^\ast}$.} $\Eq{V_t^{\alpha^\ast}\middle| \Fcal_s}=V_s^{\alpha^\ast}$ (independently on the choice of $v_0$)\footnote{Example \ref{example_fixed} shows that the definition of consistency for a fixed $v_0$ leads to a different notion}. 
\end{definition}

Moreover we make an abuse of notation, as a priori the strategies $(\alpha^\ast_u)_{u\in [0,s]}$, $(\alpha^\ast_u)_{u\in [0,t]}$ which define $V_s^{\alpha^\ast}$, $V_t^{\alpha^\ast}$ might be different on $[0,s]\subset [0,t]$. Nevertheless, assuming that consistency holds in the chosen market model, they must coincide over the common time interval.

This type of consistency naturally leads to strategies which are not affected by a fixed time horizon, which is the basis of the well known theory of forward performances, which we now recall, as it will play a crucial role in the discussion. 

\begin{definition}\label{def:fwd_performance}
    Given $u_0:\R \to \R$ a concave and increasing function, an $\F$-adapted process $(\mbu_t)_{t\geq 0}$ with $\mbu_t:\Omega\times \R\to \R$ is a forward performance if
\begin{itemize}
    \item the mapping $x \mapsto \mbu_t(\om, x)$ is increasing and concave for each $\om \in \Om$ and $t \geq0$;
    \item it satisfies $\mbu_0(\om, x) = u_0(x)$;
    \item for all $t,s \in [0,\infty)$ such that $t > s$ and for any $\alpha\in\mathcal{A}$,  $\Ep{\mbu_t(V_t^\alpha))\middle| \Fcal_s}\leq \mbu_s(V_s^\alpha)$ holds;
    \item for all $t,s \in [0,\infty)$ such that $t > s$ there exists $\alpha^{\ast}\in\mathcal{A}$ such that   
    $\Ep{\mbu_t(V_t^{\alpha^{\ast}})\middle| \Fcal_s}= \mbu_s(V_s^{\alpha^{\ast}})$.
\end{itemize}
\end{definition}

\begin{remark}
   Observe that Definition \ref{def:time_consistency} is concerned with a notion of consistency which can be effectively interpreted as a self-financing constraint for the optimal portfolio, in order to guarantee optimality at any date $t$ for a time dependent preference structure. Conversely, the concept of consistency that motivates Definition \ref{def:fwd_performance} pertains to a natural `martingale optimality  principle' characterizing optimality in stochastic control problems.
\end{remark}

\subsection{On consistency for Von Neumann-Morgenstern type preferences}

Consistency may fail already in classical optimization problems, when utilities are not state dependent, as described in the following proposition. In complete markets, the classical Merton problem can be solved by computing explicitly at time $t$ the optimal discounted wealth $\xi_t^\ast$ and then deriving optimal strategies backwards via a perfect hedging procedure, which intrinsically links them to the terminal date. It seems therefore appropriate to discuss this simple motivating example before proceeding to our main result. In fact the proposition implies that if the squared market price of risk has no deterministic representative, an agent with classical preferences may not be able to devise consistent optimal plans at all.

\begin{proposition}\label{prop:det_exp_ut}
Assume that \eqref{hp_theta} holds true for
$\theta_t=-\left(\frac{\mu_t-r}{\sigma_t}\right)$.
For
\begin{equation}\label{eq:exp_ut_deterministic}
u(x)=-\frac1\gamma e^{-\gamma x},
\qquad \gamma>0,
\end{equation}
problem \eqref{optimization} is consistent if and only if there exists a deterministic, nonnegative measurable function $a:[0,\infty)\to\R$, bounded on each finite interval, such that $\theta_t^2=a(t)$ for $dt\otimes d\PW$-almost every $(t,\omega)$. In this case,
\[
\xi_t^{*}
=
v_0-\frac1\gamma\int_0^t\theta_sdW_s^{\PWq}.
\]
\end{proposition}

\begin{remark}\label{power:utility}
The issue of inconsistency arises not only for utilities of the exponential type. Simple inspections show that the Merton problem is always consistent for $u(x)=\ln(x)$, but for power utilities $u(x)=x^\gamma/\gamma$ with $\gamma\in(0,1)$ and initial wealth $v_0>0$, problem \eqref{optimization} is consistent if and only if $\theta_t^2=a(t)$ for a deterministic measurable function $a$ (see Appendix \ref{proof:power} for the proof), in analogy with the case treated in Proposition \ref{prop:det_exp_ut}.
\end{remark}

\begin{remark}[What if consistency fails?]\label{failure}
If we consider two times $s<t$, the respective optima $V_s^{\pi^\ast}$ and $V_t^{\pi^\ast}$ and the risk neutral price
\[
\Pi_s:=\Eq{V_t^{\pi^\ast}\middle|\Fcal_s}
\]
(i.e. the hedging price at time $s$ of the optimum $V_t^{\pi^\ast}$), situations where
$\PW(\Pi_s\neq V_s^{\pi^\ast})>0$
may arise if consistency does not hold true. This issue is of little relevance whenever the agent cannot change her \textbf{initially} specified investment horizon $t$. We stress that often in practice an agent chooses to optimize up to time $t$ with the option of reconsidering her initially set horizon, either with $t_1<t$ or $t_2>t$. Usually short time horizons are considered and after that (and depending on the performance of the portfolio), the agent decides either to quit investing or to extend the final date.

In this situation once the optimum is found, then the agent performs the strategy $\pi^\ast$ up to time $s$ reaching the output $V_s^{\pi^\ast}$. However, as the event
$A=\{V_s^{\pi^\ast}<\Pi_s\}\in\Fcal_s$
may have positive probability, if $A$ occurs the agent will not be able to roll over her strategy from $s$ to $t$ in order to achieve the time $t$ optimum $V_t^{\pi^\ast}$.
\end{remark}

\section{Main results}\label{main results}

Motivated by the previous discussion, we are now ready to illustrate the main contribution of this paper. 
We consider an agent who plans to invest in the market which has a \textbf{stochastic market price of risk} $(\theta_t)_{t\geq 0}$. The agent faces uncertainty on her future risk aversion and may be willing to adjust the terminal date of the investment over time. We assume that the preferences of the agent are represented by a state-dependent utility function of the exponential type $\mbu_t:\Omega \times \R \to \R$ for $t\in [0,\infty)$, such that $$\mbu_t(\om, x) = -\frac{1}{\gamma_t(\om)}e^{-\gamma_t(\om) x}.$$ For our purposes, we model the stochastic risk aversion $(\gamma_t)_{t\geq0}$ specifying the dynamics of the general $\F$-adapted process $\big(1/\gamma_t\big)_{t\geq 0}$ as follows:
\begin{align*}
    \begin{cases}
        d\left(\frac{1}{\gamma_t}\right)=\frac{1}{\gamma_t}(\eta_t dt+\beta_t d W^{\PWq}_t), \\ \gamma_0>0,
    \end{cases}
\end{align*}
where $dW^{\PWq}_t=dW_t-\theta_tdt$ is a $\PWq$-Brownian motion. As is customary, $1/\gamma_t$ will be referred to as risk tolerance at time $t$.
We will be working with the following set of assumptions in addition to \eqref{hp_theta}:
\begin{assumption}[A]
The processes $(\eta_t)_{t\geq 0}, (\beta_t)_{t\geq 0}$ are progressively measurable and for any $t>0$ there exists a constant $K_t > 0$ such that $$\PW\left(\sup_{0\leq u\leq t}|\eta_u|<K_t,\sup_{0\leq u\leq t}|\beta_u|<K_t\right)=1.$$
\end{assumption}

Without any further mention, from now on, all coefficient identities below are understood $dt\otimes d\PW$-almost everywhere. Deterministic derivatives are understood Lebesgue-almost everywhere. 

The following result provides a full characterization of consistency.

\begin{theorem}\label{main:thm}
Suppose Assumption (A) and (1) hold. Consider $\mbu_t(\omega,x)= -\frac{1}{\gamma_t(\omega)} e^{-\gamma_t(\omega) x }$, with $(\gamma_t)_{t\geq 0}$ solving
\(d\left(\frac{1}{\gamma_t}\right)=\frac{1}{\gamma_t}(\eta_t dt+\beta_t d W^{\PWq}_t)\), \(\gamma_0>0\).
Problem (4) is consistent if and only if there exist two deterministic, locally bounded, measurable functions $e(t)$ and $f(t)$ such that $\eta_t=e(t)$ and
$$\beta_t=
-\frac{\theta_t}{2}
+\frac{f(t)+e(t)(b_t-\ln Z_t)}{\theta_t} \qquad \text{with } b_t:=\int_0^tf(s)ds .$$
Moreover, for this choice of $(\eta_t)_{t\geq 0}$ and $(\beta_t)_{t\geq 0}$ the following hold:
\begin{itemize}
    \item The optimal wealth process is given by
\begin{equation*}
\xi_t^{*}=\frac{1}{\gamma_t}(c_tv_0+b_t-\ln Z_t) \quad \text{where } \quad c_t:=\gamma_0\exp\left(-\int_0^t e(s)\,ds\right).
\end{equation*}

\item The optimal strategy in the risky asset is
$$\alpha_t^*=\frac{\beta_t}{\sigma_t}
-\frac{\theta_t}{\gamma_t\sigma_t\xi_t^*}\quad \text{on} \quad \{\xi_t^*\neq0\}.$$

\item $(\mbu_t(V_t^{\alpha^*}))_{t\geq0}$ is a $\PW$-martingale for every initial wealth if and only if $\eta_t=0$, $\beta_t=-\theta_t/2$.
\end{itemize}

\end{theorem}

\begin{example}[Martingality does not imply optimality]\label{example_fixed}
Fix $v_0\in\R$ and $\delta>0$. Let $H$ be the unique solution of
\begin{equation}\label{eq:auxiliary_proc}
dH_t
=
\left(
\frac{\delta\theta_t}{2(1+H_t^2)}
-\frac{\theta_t^2}{2}
\right)dt
-\theta_tdW_t^{\PWq},
\qquad
H_0=\gamma_0v_0,
\end{equation}
and set
\begin{equation}\label{eq:eta_beta_fwd_perf}
\eta_t
=
\frac{\delta\theta_t}{2(1+H_t^2)},
\qquad
\beta_t
=
\frac{\delta(1+H_t)}{2(1+H_t^2)}
-\frac{\theta_t}{2}.
\end{equation}
The drift in \eqref{eq:auxiliary_proc} is globally Lipschitz in $H$ on each finite horizon, and the coefficients in \eqref{eq:eta_beta_fwd_perf} satisfy Assumption (A). Define $1/\gamma$ using these coefficients and $V_t^*=H_t/\gamma_t$. Direct calculation gives
\[
dV_t^*
=
\frac1{\gamma_t}(\beta_tH_t-\theta_t)dW_t^{\PWq},
\qquad
d\mbu_t(V_t^*)
=
\mbu_t(V_t^*)(\beta_t+\theta_t)dW_t.
\]
Both are true martingales on finite horizons: for $V^*$ this follows from the finite moment estimates for $H$ and $1/\gamma$, and for the utility from boundedness of $\beta+\theta$.

Nevertheless, this is not a forward performance, since $\eta\neq0$ almost everywhere, by Proposition \ref{prop:infinite_new}. Indeed, for another wealth process $V^\pi$, put
\[
q_t
=
\gamma_t\sigma_t\pi_t-\beta_t\gamma_tV_t^\pi.
\]
The drift of $\mbu_t(V_t^\pi)$ equals
\[
\mbu_t(V_t^\pi)
\left\{
\frac12(q_t+\theta_t)^2
+\eta_t\gamma_t(V_t^\pi-V_t^*)
\right\}.
\]
The last term has no prescribed sign. Cancelling the drift at $V^*$ therefore does not establish the supermartingale inequality for other wealth processes.
Moreover,
\[
H_t
=
\gamma_0v_0+\int_0^t\eta_sds-\ln Z_t.
\]
Hence, by the first-order condition for the static problem, $V_t^*$ is its optimizer at every horizon if and only if $\int_0^t\eta_sds$ is deterministic for every $t$. The martingale property of utility along $V^*$ alone does not prove optimality.

The same field is not consistent for every initial wealth: $\eta_t$ is genuinely random for each $t>0$, contradicting Theorem \ref{main:thm}. Once $v_0$ is fixed in this construction, the utility field is kept unchanged when comparing investment strategies.
\end{example}

\begin{example}[Consistency for a fixed initial wealth does not imply
consistency for every initial wealth]
Fix $v_0\in\mathbb R$, $\gamma_0>0$, and $\delta>0$.
Consider a market with $\theta_t=1$, so that
\(Z_t=\exp\left(W_t-\frac{t}{2}\right)\),
and $W_t^{\PWq}=W_t-t$. Set
\[
L_t:=\ln Z_t-\gamma_0v_0,
\qquad
\eta_t:=\frac{\delta}{1+L_t^2},
\qquad
\beta_t:=-\frac12-\frac{\delta L_t}{1+L_t^2},
\]
and define the risk tolerance by
\[
d\left(\frac1{\gamma_t}\right)
=
\frac1{\gamma_t}
\left(\eta_t\,dt+\beta_t\,dW_t^Q\right),
\qquad
\frac1{\gamma_0}>0.
\]
Both coefficients are bounded, and hence satisfy Assumption (A).
For the utility field \(\mbu_t(x)=-\frac1{\gamma_t}e^{-\gamma_t x}\) the  optimal wealth is \(\xi_t^*:=-\frac{L_t}{\gamma_t}\) and $\xi_0^*=v_0$. Under $\PWq$
\[
dL_t=\frac12\,dt+dW_t^Q.
\]
The product rule therefore gives
\begin{equation*}
d\xi_t^* = -\frac1{\gamma_t}
\left[
\left(\eta_tL_t+\frac12+\beta_t\right)dt
+
(1+\beta_tL_t)\,dW_t^{\PWq}
\right] =
-\frac{1+\beta_tL_t}{\gamma_t}\,dW_t^{\PWq}.
\end{equation*}

Boundedness of $\eta$ and $\beta$ implies that $1/\gamma$
has finite moments of every positive order on finite horizons.
The same holds for the running supremum of $|L|$.
Consequently, $\xi^*$ is a square-integrable $\PWq$-martingale on every finite horizon, and
$\mathbb E_Q[\xi_t^*]=v_0$. Since $\xi^*$ is a $\PWq$-martingale, the problem is consistent
for the prescribed initial wealth $v_0$. However, for every $t>0$, the random variable $L_t$ has a
nondegenerate Gaussian distribution. Hence
$\eta_t=\delta/(1+L_t^2)$ is genuinely random.
Theorem 3.1 therefore excludes consistency for every initial wealth.
\end{example}

\begin{remark}[Martingale risk tolerance]\label{rem:implications}
If $\eta_t=0$, consistency is equivalent to
\[
\theta_t\left(\frac{\theta_t}{2}+\beta_t\right)=a(t)
\]
for a deterministic locally bounded measurable function $a$. In particular,
\[
\beta_t
=
-\frac{\theta_t}{2}+\frac{a(t)}{\theta_t},
\qquad
\xi_t^{*}
=
\frac1{\gamma_t}
\left(
\gamma_0 v_0+\int_0^ta(s)ds-\ln Z_t
\right).
\]
If $|\theta_t|$ is bounded away from zero on each finite interval, every locally bounded $a$ gives bounded coefficients. Without such a lower bound, the quotient $a(t)/\theta_t$ might explodes to infinity, so that $\beta_t$ violates Assumption (A). Morover for any non zero $a(t)$ the stochastic utility fails to be forward performance.

For the particular choice $a=0$, risk tolerance and risk aversion have dynamics under $\PW$
\[
d(1/\gamma_t)
=
\frac1{\gamma_t}
\left(
\frac{\theta_t^2}{2}dt-\frac{\theta_t}{2}dW_t
\right),
\qquad
d\gamma_t
=
\gamma_t
\left(
-\frac{\theta_t^2}{4}dt+\frac{\theta_t}{2}dW_t
\right).
\]
We observe that $(\gamma_t)_{t\geq 0}$ is a positive supermartingale and hence there exists a random variable $\gamma_\infty \in L^1(\Om, \Fcal, \PW)$ such that $\gamma_t \to \gamma_\infty$ $\PW$-a.s. as $t \to \infty$ see \cite{RY99} Corollary II.2.11. In particular, if $\PW(\gamma_\infty \neq 0)=1$ then one obtains $1/\gamma_t \to 1/\gamma_\infty$ $\PW$-a.s. as well. One can therefore think of $1/\gamma_\infty$ as defining the agent's unknown true preferences $(\mbu, \PW)$ with $\mbu(\om, x) = -\frac{1}{\gamma_\infty(\om)}e^{-\gamma_\infty(\om) x}$. Since $\big(1/\gamma_s\big)_{0\leq s \leq t}$ is a $\PWq$-martingale for all $t$, for every pair $s<t$, the process $1/\gamma_s$ can be seen as the prediction at time $s$ of the future risk tolerance $1/\gamma_t$. If, moreover, there existed an equivalent probability measure\footnote{Although intuitively appealing, this assertion requires  a stronger integrability assumption on $\left(\theta_t\right)_{t\geq 0}$ in order to ensure the existence of a suitable equivalent probability measure $\tilde{\PWq}$ on $(\Om, \Fcal)$ and $(1/\gamma_t)_{t\geq 0}$ uniformly integrable. Under this assumptions the conclusion follows from a straightforward application of martingale convergence Theorem (see \cite{RY99} Theorem II.3.1).} $\tilde{\PWq}$ on $\Fcal$ we could interpret $1/\gamma_t$ as the `most rational' prediction of the real risk tolerance $1/\gamma_\infty\in L^{1}(\Omega,\Fcal,\PWq)$, through the relation $\frac{1}{\gamma_t} = \Eq{\frac{1}{\gamma_\infty}|\Fcal_t}$. Therefore, in devising her investment strategy, the agent is adopting a consistent estimate of her true risk aversion given the information available on the market. The novelty of this approach lies in the fact that the longer the agent engages with the market, the more she becomes aware of her true risk aversion This leads to a strong interplay between the time evolution of the agent's strategy and the updating of her preference order.\end{remark}

\begin{proposition}\label{prop:infinite_new}
Under \eqref{hp_theta} and Assumption (A),
$\mbu_t(\omega,x)= -\frac{1}{\gamma_t(\omega)} e^{-\gamma_t(\omega) x }$
is a forward performance in the sense of Definition \ref{def:fwd_performance} if and only if
\[
\eta_t=0,
\qquad
\beta_t=-\frac{\theta_t}{2}.
\]
Its optimal wealth is
\[
\xi_t^{*}
=
\gamma_t^{-1}(\gamma_0 v_0-\ln Z_t).
\]
and
\begin{equation}\label{eq:utility_at_opt}
\mbu_t(\xi_t^{*})
=
-e^{-\gamma_0v_0}\frac{Z_t}{\gamma_t}.
\end{equation}
\end{proposition}

Consistency guarantees that the optimal terminal wealths form one self-financing process. It does not imply that the expected optimal utility is constant across horizons. Conversely, a forward performance necessarily provides the same optimal wealth at each horizon, under our admissibility convention, and is therefore consistent for every initial wealth.

For deterministic market price of risk, consistency for every initial wealth has a simpler form. The qualification concerning initial wealth is essential, as the preceding example demonstrates.

\begin{proposition}\label{prop:tutti_beta}
Suppose \eqref{hp_theta} and Assumption (A) hold and $\theta_t=g(t)$ for a deterministic locally bounded measurable function $g$. The stochastic utility
$\mbu_t(x)=-\gamma_t^{-1}e^{-\gamma_tx}$
is consistent  if and only if \(\eta_t=0\), \(\beta_t=f(t)\)
for a deterministic locally bounded measurable function $f$. Within this family it is a forward performance if and only if $f(t)=-g(t)/2$ almost everywhere.
\end{proposition}

\begin{example}[Multiplicative stochastic noise]
We conclude by mentioning an alternative class of stochastic utilities, where a deterministic utility is perturbed by a multiplicative martingale noise. While the expected utility of the agent on deterministic quantities is unaffected (and so are her preferences), the randomness of the market interacts with the martingale noise, distorting the agent's perceived utility from a portfolio. We find, however, that the problem is consistent only if the parameter of the noise process suitably counterbalances the randomness of the market. Additionally requiring that $\mbu_t(\om, x)$ is a forward performance results in the agent not investing at all in the risky asset. 
\begin{proposition}\label{prop:ut_w_noise}
Suppose \eqref{hp_theta} holds. Let
\[
\mbu_t(x)
=
-\frac{X_t}{\gamma}e^{-\gamma x},
\qquad
dX_t=X_t\beta_tdW_t,
\quad X_0=1,
\quad \gamma>0,
\]
where $(\beta_t)_{t\geq 0}$ is progressively measurable and for any $t>0$ there exists a constant $K_t > 0$ such that $\PW\left(\sup_{0\leq u\leq t}|\beta_u|<K_t\right)=1$. Problem \eqref{optimization} is consistent if and only if there is a deterministic nonnegative measurable function $k$ such that $(\theta_t-\beta_t)^2=k(t)$.
In that case,
\[
\xi_t^{*}
=
v_0-\frac1\gamma\int_0^t(\theta_s-\beta_s)dW_s^{\PWq},
\qquad
\alpha_t^{*}\xi_t^*=
-\frac{\theta_t-\beta_t}{\gamma\sigma_t}.
\]
$\mbu_t$ is a forward performance if and only if $\beta_t=\theta_t$ almost everywhere, in which case the optimal discounted risky amount is zero and the optimal wealth is constantly $v_0$.
\end{proposition}
\end{example}

Questa versione distingue la consistenza dall’indifferenza rispetto alla data di uscita e formula il risultato per tempi di arresto limitati.

\begin{remark}
The results described so far naturally lead to the question of whether an optimal random time exists to exit the investment plan. Unfortunately consistency alone does not imply that the agent is indifferent between different investment horizons, as the expected optimal utility may depend on the terminal date. Nevertheless for the choice
\(\eta_t=0\) and $\beta_t=-\frac{\theta_t}{2}$,
the forward performance property offers the advantage of ensuring no regret, regardless of the exit time chosen by the agent.
\\ Consider the family $\mathcal{T}$ of $\PW$ almost surely finite stopping times $\tau:\Omega\to (0,+\infty)$, then the process $\xi^*_{t\wedge \tau}$ is a $\PWq$-martingale for every $\tau\in \mathcal{T}$ and it optimizes $\Ep{\mbu_{t\wedge\tau}(\cdot)}$. Under some technical assumption it is possible to infer that the agent is indifferent among all the possible terminal dates, as the expected reward of optimizing over all possible stopping rules does not exceed the one obtained without stopping, i.e. 
\[\sup_{\tau\in\mathcal{T}} \Ep{\mbu_{\tau}(\xi^*_{\tau})}= \Ep{\mbu_{t}(\xi^*_t)} \quad \forall\,t>0.\]
\end{remark}

\appendix

\section{Proofs}

All arguments are carried out on arbitrary finite horizons with the corresponding pricing measure. We first prove Proposition \ref{prop:det_exp_ut}, then the general characterization and its martingale-risk-tolerance specialization. We conclude with the forward performance and deterministic-market results and the multiplicative-noise example.

\subsection{Proof of Proposition \ref{prop:det_exp_ut}}

Consider $t>0$ fixed. Since the market is assumed to be complete, we can recast Problem \eqref{optimization} as the following infinite dimensional problem
    \begin{equation}\label{static}
        \sup \left\{\Ep{u(\xi)}\mid \xi \in L^1(\Omega,\Fcal_t,\PWq) \text{ and } \Eq{\xi} = v_0\right\}.
    \end{equation}  
The application of Lemma \ref{utile} then yields the optimum
\begin{equation*}
    \xi^\ast_t = v_0 + \frac{1}{\gamma}\Eq{\ln(Z_t)} - \frac1\gamma \ln(Z_t),
\end{equation*}
with $\xi^\ast_t\in L^1(\Omega,\Fcal_t,\PWq)$, as a consequence of \eqref{hp_theta}.
Considering some time $s < t$ and mimicking the argument in \cite{Bjork09}, Proposition 20.11, we find
\begin{equation*}\label{eq:price_claim_exp}
    \Eq{\xi_t^*\middle| \Fcal_s} = v_0 + \frac{1}{\gamma} \left[E_0^t - E_s - \ln(Z_s)\right],
\end{equation*}
where $E_0^t = \Eq{\frac12 \int_0^t \theta_u^2 du}$ and $E_s = \Eq{\frac12 \int_s^t \theta_u^2 du\middle|\Fcal_s}$. In order to ensure the problem is time consistent we require $\Eq{\xi_t^*\middle| \Fcal_s}=\xi^*_s$, i.e. 
\begin{equation}\label{eq:exp_timeconsistency_cond}
    \frac{1}{\gamma} \left[E_0^t - E_s - \ln(Z_s)\right] =  \frac1\gamma\Eq{\ln(Z_s)} - \frac1\gamma \ln(Z_s),
\end{equation}
where the dynamics of process $\ln(Z_t)$ under $\PWq$ is given by $d\ln (Z_t)= \frac{1}{2}\theta_t^2dt+\theta_t dW^{\PWq}_t$.
Observing that 
\begin{align*}
    E_0^t = \Eq{\frac12 \int_0^t \theta_u^2 du} = \underbrace{\Eq{\frac12 \int_0^s \theta_u^2 du}}_{\Eq{\ln(Z_s)}} + \Eq{\frac12 \int_s^t \theta_u^2 du}, 
\end{align*}
Eq. \eqref{eq:exp_timeconsistency_cond} boils down to 
\begin{equation}\label{eq:indep}
    \Eq{\frac12 \int_s^t \theta_u^2 du\middle|\Fcal_s} = \Eq{\frac12 \int_s^t \theta_u^2 du}.
\end{equation}
 Because the choice of $s$ is arbitrary, condition \eqref{eq:indep} must hold for all $s < t$. Whenever $t\mapsto\theta^2_t$ is deterministic the latter is clearly satisfied. We conclude by showing the converse implication. First, applying Fubini-Tonelli Theorem we can rewrite \eqref{eq:indep} as 
 \begin{equation*}
     \frac12 \int_s^t\Eq{\theta^2_u \middle|\Fcal_s} du = \frac12 \int_s^t\Eq{\theta^2_u} du,
 \end{equation*}
 which in turn implies that $\Eq{\theta^2_u \middle|\Fcal_s} = \Eq{\theta^2_u}$ for almost every $u \in (s, t)$. Furthermore, for $u < t$ fixed we also have that $\Eq{\theta^2_u \middle|\Fcal_s} = \Eq{\theta^2_u}$ for almost every $0 \leq s < u$. Define the stochastic process $I(s) := \Eq{\theta^2_u \middle|\Fcal_s}$ for $0 \leq s < u$. Since \eqref{hp_theta} holds true, $\theta^2_u \in L^1(\Om, \Fcal_u, \PWq)$, then $I(s)$ is a martingale that admits a representation of the form $I(s) = I_0 + \int_0^s \psi_t dW^\PWq_t$ where $\psi_t$ is a predictable, square integrable process and $I_0 \in \R$. Without loss of generality we consider the continuous version of $I(s)$ and hence, letting $s \to u$ we obtain $I(s) \to I(u) = \Eq{\theta^2_u \middle| \Fcal_u} = \theta^2_u$. Recalling that $\Eq{\theta_u^2\middle|\Fcal_s} = \Eq{\theta_u^2}$ for all $s < u$ we conclude that $\Eq{\theta_u^2} = \theta^2_u$ and, since the choice of $u \in (s,t)$ was arbitrary, the function $\theta_t^2=a(t)$ must be a deterministic function $dt\otimes d\PWq\text{-almost everywhere}$. Simple computations show $
\xi_t^*=v_0-\frac1\gamma\int_0^t\theta_udW_u^{\PWq}.$

Conversely, suppose that $\theta_t^2=a(t)$
$dt\otimes d\mathbb{P}$-almost everywhere for such a function $a$. Then
\[
\xi_t^*
=
v_0-\frac1\gamma\int_0^t\theta_udW_u^{\PWq}.
\]
This is a square-integrable $\PWq$-martingale on every finite horizon, which proves consistency.

 \hfill \qedsymbol

\subsection{Proof of Theorem \ref{main:thm}}
Fix $t>0$, and consider Problem \eqref{optimization} for the state-dependent utility function $\mbu_t(\om, x) = -\frac{1}{\gamma_t(\om)}e^{-\gamma_t(\om) x}$. We reformulate the maximization problem as 
    \begin{equation*}
        \sup \left\{\Ep{\mbu_t(\xi)}\mid \xi \in L^1(\Omega,\Fcal_t,\PWq) \text{ and } \Eq{\xi} = v_0\right\}.
    \end{equation*} 
Applying Lemma \ref{utile} to $\mbu_t(\om, x)$ yields the $t$-optimal discounted profile
\begin{equation}\label{eq:opt_profile_sdu}
    \xi^\ast_t = -\frac{1}{\gamma_t}\ln\left(\lambda Z_t\right).
\end{equation}
One can then recover $\lambda$ by plugging \eqref{eq:opt_profile_sdu} into the budget constraint $\Eq{\xi^\ast_t} = v_0$ 
\begin{equation*}
    \ln(\lambda)=\frac{1}{\Eq{1/\gamma_t}}\left(-v_0-\Eq{\frac{1}{\gamma_t}\ln (Z_t)}\right).
\end{equation*}
Substituting in \eqref{eq:opt_profile_sdu} and setting $c_t = \frac{1}{\Eq{1/\gamma_t}}$ we finally obtain
\begin{equation}\label{eq:opt_profile_sdu_final}
    \xi^\ast_t = \frac{1}{\gamma_t} \left\{c_t\left(v_0+ \Eq{\frac{1}{\gamma_t}\ln (Z_t)}\right)  -\ln (Z_t)\right\}.
\end{equation}
Notice that $\xi_t^*\in L^1(\Omega,\Fcal,\PWq)$ since Assumptions (A) and \eqref{hp_theta} imply that, for every
$T<\infty$ and every finite $p\geq1$,
\[
    \Eq{\sup_{0\leq s\leq T}
        \left(\frac{1}{\gamma_s}\right)^p}<\infty,
    \qquad
    \Eq{\sup_{0\leq s\leq T}|\ln(Z_s)|^p}<\infty.
\]

Consider now some $s<t$ and denote by $\xi^\ast_s$ the optimal discounted profile for the investment horizon $s$, the problem is consistent whenever the condition $\xi^\ast_s = \Eq{\xi^\ast_t\middle|\Fcal_s}$ is verified. 

By assumption $d\left(\frac{1}{\gamma_t}\right)=\frac{1}{\gamma_t}(\eta_t dt+\beta_t dW^{\PWq}_t)$
for $\eta_t,\beta_t$ $\F$-adapted processes and with $W^{\PWq}_t$ a $\PWq$-Brownian motion defined by $dW^{\PWq}_t=dW_t-\theta_tdt$. 
Observing that $d\ln (Z_t)= \frac{1}{2}\theta_t^2dt+\theta_t dW^{\PWq}_t$, stochastic Leibniz rule yields 
\begin{eqnarray}\label{eq:dyn_gamma_logz}
   d\left(\frac{1}{\gamma_t}\ln (Z_t)\right) & = & \frac{1}{\gamma_t}\left(\frac{\theta_t^2}{2}+\eta_t\ln (Z_t)+\theta_t\beta_t\right)dt + \frac{1}{\gamma_t}\left(\theta_t+\beta_t \ln (Z_t) \right) dW^{\PWq}_t.
\end{eqnarray} 
Recall that $\theta, \eta, \beta$ satisfy the boundedness conditions in Assumption (A). Then in particular we can rewrite the latter in integral form and take the expectation conditional on $\Fcal_s$ to obtain  
\begin{equation}\label{cond:entropy}
\Eq{\frac{1}{\gamma_t}\ln (Z_t)\middle | \Fcal_s}= \frac{1}{\gamma_s}\ln (Z_s)+\Eq{\int_s^t \frac{1}{\gamma_u}\left(\frac{\theta_u^2}{2}+\eta_u\ln(Z_u)+\theta_u\beta_u\right)du \middle| \Fcal_s}.
\end{equation}
Plugging \eqref{eq:opt_profile_sdu_final} into the time consistency condition we obtain
\begin{eqnarray}\label{costanti} 
& & c_t\left(v_0+\Eq{\frac{1}{\gamma_t}\ln(Z_t)}\right)\Eq{\frac{1}{\gamma_t}\middle | \Fcal_s}-\Eq{\frac{1}{\gamma_t}\ln (Z_t)\middle | \Fcal_s}
\\ \nonumber & = & c_s\left(v_0+\Eq{\frac{1}{\gamma_s}\ln (Z_s)}\right)\frac{1}{\gamma_s}-\frac{1}{\gamma_s}\ln (Z_s).
\end{eqnarray}
Let us set $k_t = c_t\left(v_0+\Eq{\frac{1}{\gamma_t}\ln (Z_t)}\right)$ and observing that 
\begin{equation*}
    \Eq{\frac{1}{\gamma_t}\middle| \Fcal_s}=\frac{1}{\gamma_s}+\Eq{\int_s^t \frac{1}{\gamma_u}\eta_u du \middle| \Fcal_s},
\end{equation*}
Eq. \eqref{costanti} boils down to 
\begin{equation*}
    \Eq{\frac{1}{\gamma_t}\ln (Z_t)\middle | \Fcal_s}-\frac{1}{\gamma_s}\ln (Z_s) = (k_t-k_s)\frac{1}{\gamma_s}+k_t\Eq{\int_s^t \frac{1}{\gamma_u}\eta_u du \middle| \Fcal_s}.
\end{equation*}
Substituting \eqref{cond:entropy} in the latter we find
\begin{equation*}
    \Eq{\int_s^t \frac{1}{\gamma_u}\left(\frac{\theta_u^2}{2}+\eta_u\ln (Z_u)+\theta_u\beta_u\right)du \middle| \Fcal_s}= (k_t-k_s)\frac{1}{\gamma_s}+k_t\Eq{\int_s^t \frac{1}{\gamma_u}\eta_u du \middle| \Fcal_s},
\end{equation*}
and a simple algebraic rearrangement then provides
\begin{equation}\label{relazione:importante}\frac{k_t-k_s}{\gamma_s}=\Eq{\int_s^t \frac{1}{\gamma_u}\left(\frac{\theta_u^2}{2}+\eta_u\ln (Z_u)+\theta_u\beta_u-k_t\eta_u\right)du \middle| \Fcal_s},
\end{equation}
which has to hold for any couple $s,t \in [0,\infty)$ such that $s<t$.

\medskip
 
For $s<t$ we define the following $\PWq$-martingale:
\begin{equation*}
    M_s := \Eq{\int_0^t \frac{1}{\gamma_u}\left(\frac{\theta_u^2}{2}+\eta_u\ln (Z_u)+\theta_u\beta_u-k_t\eta_u\right)du \middle|\Fcal_s}.
\end{equation*}
In light of Eq. \eqref{relazione:importante} we may write
\begin{equation}\label{eq:q_mart2}
    M_s = \int_0^s \frac{1}{\gamma_u}\left(\frac{\theta_u^2}{2}+\eta_u\ln (Z_u)+\theta_u\beta_u-k_t\eta_u\right)du + \frac{k_t-k_s}{\gamma_s}.
\end{equation}
We first justify differentiation of the deterministic functions
appearing above. Recall that Assumptions (A) and \eqref{hp_theta} imply that, for every
$T<\infty$ and every finite $p\geq1$,
\[
    \Eq{\sup_{0\leq s\leq T}
        \left(\frac{1}{\gamma_s}\right)^p}<\infty,
    \qquad
    \Eq{\sup_{0\leq s\leq T}|\ln(Z_s)|^p}<\infty.
\]
In particular, taking expectations in the dynamics of
$1/\gamma_s$ and $(1/\gamma_s)\ln(Z_s)$ gives
\begin{align*}
    \frac{d}{ds}\Eq{\frac{1}{\gamma_s}}
    &=
    \Eq{\frac{\eta_s}{\gamma_s}},
    \\
    \frac{d}{ds}\Eq{\frac{1}{\gamma_s}\ln(Z_s)}
    &=
    \Eq{
        \frac{1}{\gamma_s}
        \left(
            \frac{\theta_s^2}{2}
            +\eta_s\ln(Z_s)
            +\theta_s\beta_s
        \right)
    }
\end{align*}
for Lebesgue-almost every $s$.
Both functions are absolutely continuous with essentially
bounded derivatives on every finite interval.
Moreover, $\Eq{1/\gamma_s}$ is bounded away from zero on
each such interval. Indeed, if $|\eta_s|\leq K_T$ on $[0,T]$,
the stochastic exponential representation of $1/\gamma_s$
and boundedness of $\beta$ yield
\[
    \Eq{\frac{1}{\gamma_s}}
    \geq \frac{1}{\gamma_0}e^{-K_TT},
    \qquad 0\leq s\leq T.
\]
Consequently, $c$ and $k$ are locally absolutely continuous,
with locally essentially bounded derivatives.

Suppose now that consistency holds for every initial wealth.
Since $k_t$ is constant with respect to $s$, the product rule
gives
\[
    d\left(\frac{k_t-k_s}{\gamma_s}\right)
    =
    \frac{k_t-k_s}{\gamma_s}
    \left(\eta_sds+\beta_sdW_s^{\PWq}\right)
    -\frac{\dot k_s}{\gamma_s}ds.
\]
Substituting into \eqref{eq:q_mart2} and collecting terms,
we obtain
\begin{align*}
    dM_s
    ={}&
    \frac{1}{\gamma_s}
    \left[
        \frac{\theta_s^2}{2}
        +\theta_s\beta_s
        +\eta_s\bigl(\ln(Z_s)-k_s\bigr)
        -\dot k_s
    \right]ds
    \\
    &+
    \frac{k_t-k_s}{\gamma_s}\beta_s\,dW_s^{\PWq}.
\end{align*}
Since $M$ is a martingale, its drift vanishes. As the terminal
horizon $t$ is arbitrary, this implies
\begin{equation}\label{eq:corrected_consistency}
    \dot k_s
    =
    \frac{\theta_s^2}{2}
    +\theta_s\beta_s
    +\eta_s\bigl(\ln(Z_s)-k_s\bigr)
\end{equation}
for $ds\otimes d\PW$-almost every $(s,\omega)$, using
equivalence of the measures on each finite horizon.

Define the deterministic function
\[
    b_s
    :=
    c_s\Eq{\frac{1}{\gamma_s}\ln(Z_s)}.
\]
Then $k_s=c_sx+b_s$, with $c_0=\gamma_0$ and $b_0=0$.
Equation \eqref{eq:corrected_consistency} becomes
\begin{equation}\label{eq:corrected_affine}
    \dot c_sx+\dot b_s
    =
    \frac{\theta_s^2}{2}
    +\theta_s\beta_s
    +\eta_s\bigl(\ln(Z_s)-c_sx-b_s\bigr).
\end{equation}
Applying this identity to two distinct initial wealths,
with the same utility field, and subtracting gives
\[
    \dot c_s=-\eta_sc_s.
\]
Since $c_s$ is deterministic and strictly positive, it follows
that $\eta$ has a deterministic representative
\[
    \eta_s=e(s):=-\frac{\dot c_s}{c_s}.
\]
Assumption (A) implies that $e$ is locally essentially bounded.
Solving the preceding differential equation yields
\[
    c_s
    =
    \gamma_0\exp\left(-\int_0^s e(u)\,du\right).
\]
Substituting $\dot c_s=-e(s)c_s$ into
\eqref{eq:corrected_affine}, we obtain
\[
    \dot b_s
    =
    \frac{\theta_s^2}{2}
    +\theta_s\beta_s
    +e(s)\bigl(\ln(Z_s)-b_s\bigr).
\]
The function $b$ is locally absolutely continuous with locally
essentially bounded derivative. Set $f(s):=\dot b_s$, choosing
a measurable representative. Since $b_0=0$,
\[
    b_s=\int_0^s f(u)\,du.
\]
Finally, $\theta_s\neq 0$ almost everywhere by \eqref{hp_theta}, so
\[
    \beta_s
    =
    -\frac{\theta_s}{2}
    +\frac{f(s)+e(s)(b_s-\ln(Z_s))}{\theta_s}.
\]
This proves necessity.

Conversely, let $e$ and $f$ be deterministic measurable
functions, essentially bounded on every finite interval,
and set
\[
    b_t=\int_0^t f(s)\,ds,
    \qquad
    c_t=\gamma_0
    \exp\left(-\int_0^t e(s)\,ds\right).
\]
Suppose that
\[
    \eta_t=e(t),
    \qquad
    \beta_t
    =
    -\frac{\theta_t}{2}
    +\frac{f(t)+e(t)(b_t-\ln(Z_t))}{\theta_t},
\]
and that these coefficients satisfy Assumption (A).
Fix $v_0\in\R$ and define
\[
    k_t=c_tv_0+b_t,
    \qquad
    \xi_t^*=\frac{1}{\gamma_t}\bigl(k_t-\ln(Z_t)\bigr).
\]
Since $\dot c_t=-e(t)c_t$ and $\dot b_t=f(t)$,
the prescribed coefficients imply
\[
    \dot k_t
    =
    \frac{\theta_t^2}{2}
    +\theta_t\beta_t
    +\eta_t\bigl(\ln(Z_t)-k_t\bigr).
\]
Using
$d\ln(Z_t)=\frac12\theta_t^2dt+\theta_tdW_t^{\PWq}$,
It\^o's formula therefore gives
\begin{align*}
    d\xi_t^*
    ={}&
    \frac{1}{\gamma_t}
    \left[
        \dot k_t-\frac{\theta_t^2}{2}
        -\theta_t\beta_t
        +\eta_t\bigl(k_t-\ln(Z_t)\bigr)
    \right]dt
    \\
    &+
    \frac{1}{\gamma_t}
    \left[
        \beta_t\bigl(k_t-\ln(Z_t)\bigr)-\theta_t
    \right]dW_t^{\PWq}
    \\
    ={}&
    \frac{1}{\gamma_t}
    \left[
        \beta_t\bigl(k_t-\ln(Z_t)\bigr)-\theta_t
    \right]dW_t^{\PWq}.
\end{align*}
The moment estimates above imply, for every $T<\infty$,
\[
    \Eq{
        \int_0^T\frac{1}{\gamma_t^2}
        \left|
            \beta_t\bigl(k_t-\ln(Z_t)\bigr)-\theta_t
        \right|^2dt
    }<\infty.
\]
Thus $(\xi^*)_{t\geq 0}$ is a square-integrable $\PWq$-martingale on
every finite horizon. Since $\xi_0^*=v_0$, it satisfies
$\Eq{\xi_t^*}=v_0$.


\medskip

We now compute the strategy replicating the optimal profile
\[
    \xi_t^*
    =
    \frac{1}{\gamma_t}
    \left(c_tv_0+b_t-\ln(Z_t)\right).
\]
Set $k_t=c_tv_0+b_t$. By the consistency condition,
\[
    \dot k_t
    =
    \frac{\theta_t^2}{2}
    +\theta_t\beta_t
    +\eta_t\bigl(\ln(Z_t)-k_t\bigr).
\]
Recalling that
\[
    d\ln(Z_t)
    =
    \frac{\theta_t^2}{2}dt+\theta_t\,dW_t^{\PWq},
\]
the product rule yields
\begin{align*}
    d\xi_t^*
    ={}&
    \frac{1}{\gamma_t}
    \left[
        \dot k_t-\frac{\theta_t^2}{2}
        -\theta_t\beta_t
        +\eta_t\bigl(k_t-\ln(Z_t)\bigr)
    \right]dt
    \\
    &+
    \frac{1}{\gamma_t}
    \left[
        \beta_t\bigl(k_t-\ln(Z_t)\bigr)-\theta_t
    \right]dW_t^{\PWq}
    \\
    ={}&
    \left(\beta_t\xi_t^*-\frac{\theta_t}{\gamma_t}\right)
    dW_t^{\PWq}.
\end{align*}
Since
$dW_t^{\PWq}=dW_t-\theta_tdt$
and
$-\theta_t=(\mu_t-r)/\sigma_t$,
we obtain
\[
    d\xi_t^*
    =
    \frac{\beta_t\xi_t^*-\theta_t/\gamma_t}{\sigma_t}
    \left((\mu_t-r)dt+\sigma_t\,dW_t\right).
\]
Comparing this expression with the self-financing wealth
dynamics, the discounted amount invested in the risky asset is
\[
    \pi_t^*
    =
    \frac{\beta_t\xi_t^*-\theta_t/\gamma_t}{\sigma_t}.
\]
Equivalently, wherever $\xi_t^*\neq0$, the proportion invested
in the risky asset is
\[
    \alpha_t^*
    =
    \frac{\pi_t^*}{\xi_t^*}
    =
    \frac{\beta_t}{\sigma_t}
    -\frac{\theta_t}{\gamma_t\sigma_t\xi_t^*}.
\]

Finally, we show that the optimal utility process
$(\mbu_t(V_t^{\alpha^*}))_{t\geq0}$ is a $\PW$-martingale if and only if
\(\eta_t=0\), \(\beta_t=-\frac{\theta_t}{2}\). Simple inspections show that under $\PW$,
\begin{equation*}
    d\mbu_t(V_t^{\alpha^*})
    = \mbu_t(V_t^{\alpha^*})
    \left[
        \bigl(e(t)(1+c_tx)-f(t)\bigr)dt
        +(\beta_t+\theta_t)dW_t
    \right],
\end{equation*}
where we used $\eta_t=e(t)$ and
$\dot k_t=-e(t)c_tx+f(t)$. If the optimal utility is a $\PW$-martingale for every initial wealth, 
\( e(t)(1+c_tx)-f(t)=0\), for each $x\in\R$. Applying this identity to two distinct initial wealths
$x_1,x_2$ and subtracting, we obtain
\(  e(t)c_t(x_1-x_2)=0\).
Since $c_t>0$, it follows that $e(t)=0$ and hence $f(t)=0$. Thus $b_t=0$, $c_t=\gamma_0$.

Conversely, for these coefficients,
\(d\mbu_t(V_t^{\alpha^*})=
    \frac{\theta_t}{2}
    \mbu_t(V_t^{\alpha^*})\,dW_t\).
Eq. \eqref{hp_theta} implies Novikov's condition on every finite
horizon. The stochastic exponential is consequently a true
$\PW$-martingale, proving the assertion.

\subsection{A self-contained proof of Proposition \ref{prop:infinite_new}}

Observe that, given the dynamics of $\frac{1}{\gamma_t}$ under the martingale measure $\PWq$ and recalling that $dW^\PWq_t = dW_t - \theta_t dt$, we can write 
\begin{equation*}
    \begin{cases}
        d\left(\frac{1}{\gamma_t}\right) = \frac{1}{\gamma_t} \left[(\eta_t -\theta_t \beta_t)dt + \beta_t dW_t\right] \\
        \gamma_0 > 0
    \end{cases}  
\end{equation*}
and consequently we have that $d\gamma_t = -\gamma_t\left[(\eta_t - \theta_t \beta_t - \beta_t^2 )dt + \beta_t dW_t\right]$. Let the portfolio process $V^\alpha_t$ be as in Eq.  \eqref{discounted:port} and observe that
$$
d\left(-\gamma_t V_t^\alpha\right) = \gamma_t V_t^\alpha \left\{(\theta_t \sigma_t \alpha_t + \eta_t -\theta_t \beta_t - \beta_t^2 +\beta_t \sigma_t \alpha_t )dt - (\sigma_t \alpha_t - \beta_t )dW_t\right].
$$
Applying Ito's formula to $\exp(-\gamma_t V_t^\alpha)$ then yields
$$
d\left(e^{-\gamma_t V_t^\alpha}\right) = e^{-\gamma_t V_t^\alpha}\gamma_t V_t^\alpha\left\{\left[\left(\sigma_t \alpha_t -\beta_t\right)\left(\theta_t + \beta_t + \frac12 \gamma_t V_t^\alpha (\sigma_t \alpha_t - \beta_t)\right) + \eta_t \right]dt - (\sigma_t \alpha_t - \beta_t )dW_t\right\}
$$
Collecting the results above, the stochastic differential of $\mbu_t(V^\alpha_t)= -\frac{1}{\gamma_t} e^{-\gamma_t V^\alpha_t}$ reads
\begin{align}\label{eq:fwd_p_dynamics}
    d\left(\mbu_t(V^\alpha_t)\right) = -\frac{1}{\gamma_t}e^{-\gamma_t V^\alpha_t} \bigg\{\left[\gamma_t V^\alpha_t (\sigma_t \alpha_t - \beta_t)\left(\theta_t + \frac12 \gamma_t V^\alpha_t(\sigma_t \alpha_t - \beta_t)\right) + \eta_t(\gamma_t V^\alpha_t + 1) -\theta_t \beta_t\right]dt + \nonumber \\
    + \left[\beta_t -\gamma_tV^\alpha_t \left(\sigma_t \alpha_t - \beta_t\right)\right]dW_t \bigg\}. 
\end{align}
To show that $\mbu_t(\om, x)$ is a forward performance, we need to impose conditions on $(\eta_t,\beta_t)$ to ensure that $\mbu_t(\om, V_t^{\alpha})$ is a supermartingale for all $\alpha \in \Acal$. Considering that $-\frac{1}{\gamma_t}e^{-\gamma_t V^\alpha_t} < 0$ for any $\alpha$, the latter condition is equivalent to requiring 
\begin{equation}\label{eq:downward_parabola}
    \alpha_t^2\left(\frac12 (\gamma_t V^\alpha_t)^2 \sigma_t^2\right) + \alpha_t \left(\gamma_t V^\alpha_t \theta_t \sigma_t - (\gamma_t V^\alpha_t)^2 \sigma_t \beta_t \right) + \frac12 (\gamma_t V^\alpha_t)^2 \beta_t^2 - \gamma_t V^\alpha_t\theta_t \beta_t + \eta_t(\gamma_t V^\alpha_t + 1) - \theta_t \beta_t \geq 0.
\end{equation}
We consider the paraboloid $(x,y)\mapsto f(x,y,\eta_t)$ defined as 
\[f(x,y,\eta_t)=x^2\left(\frac12 (\gamma_t y)^2 \sigma^2_t\right) + x \left(\gamma_t y \theta_t \sigma_t - (\gamma_t y)^2 \sigma_t \beta_t \right) + \frac12 (\gamma_t y)^2 \beta_t^2 - \gamma_t y\theta_t \beta_t + \eta_t(\gamma_t y + 1) - \theta_t \beta_t.\]
Imposing $f(x,y,\eta_t)\geq 0$ for every $x,y\in\R$ gives necessarily $\eta_t=0$ and additionally $$\theta_t(\theta_t+2\beta_t)\leq 0.$$ 
Setting 
$$h_t^{\alpha}:=
\gamma_t V_t^{\alpha}(\sigma_t\alpha_t-\beta_t)+\theta_t,$$ 
the drift of $\mbu_t(V_t^{\alpha})$ is \(\frac12\mbu_t(V_t^\alpha)
\left[(h_t^\alpha)^2-\theta_t(\theta_t+2\beta_t)\right]\). The martingale requirement along an optimal wealth process
additionally implies
\((h_t^{\alpha^*})^2
  =\theta_t(\theta_t+2\beta_t)\).
Since the left-hand side is nonnegative and the right-hand
side is nonpositive, both must vanish. By (1), this yields
$\beta_t=-\theta_t/2$ and $$\alpha_t^*=\beta_t/\sigma_t-\theta_t/(\gamma_t\sigma_t\xi_t^*)\quad \text{on} \quad \{\xi_t^*\neq 0\}.$$
For the converse, the choice
$\eta_t=0$ and $\beta_t=-\theta_t/2$, gives the optimal strategy displayed in Theorem \ref{main:thm}
\(\xi_t^*
    =
    \frac{1}{\gamma_t}
    \left(\gamma_0x-\ln(Z_t)\right)\)
which provides a forward performance.

\subsection{Proof of Proposition \ref{prop:tutti_beta}}

\begin{proof}
If $\eta_t=0$ and $\beta_t=f(t)$ is deterministic, the product
$g(t)(g(t)/2+f(t))$
is deterministic and locally bounded. Theorem \ref{main:thm} and Remark \ref{rem:implications} gives consistency.

Conversely, Theorem \ref{main:thm} gives deterministic $e,b$ with
\[
g(t)\beta_t
=
-\frac{g(t)^2}{2}
+\dot b_t
+e(t)b_t
-e(t)\ln Z_t.
\]
For almost every fixed $t>0$, the left-hand side is essentially bounded, whereas under $\PWq_t$, $\ln Z_t$ is Gaussian with variance
$\int_0^tg(s)^2ds>0$.
Thus it has unbounded support, and the identity forces $e(t)=0$.
It follows that $\eta=0$ and
\[
\beta_t
=
-\frac{g(t)}2+\frac{\dot b_t}{g(t)}
\]
has a deterministic representative $f$. Assumption (A) gives the asserted essential bounds. The forward performance assertion follows from Proposition \ref{prop:infinite_new}.
\end{proof}

\subsection{Proof of Proposition \ref{prop:ut_w_noise}}
First, we observe that the process $X_t$, defined as a solution to $dX_t = X_t \beta_t dW_t$, $X_0=1$, defines a new probability measure $\PW^\ast$ such that $X_t = \frac{d\PW^\ast}{d\PW}$. Consequently, fixing some investment horizon $t >0$ and considering $\mbu_t(\om, x) = u(x) \cdot X_t(\om)$ we have that
\begin{equation}\label{eq:ch_of_measure}
    \sup_{\substack{\alpha \in \Acal \\ V^\alpha_0 = v_0}} \Ep{\mbu_t(V_t^\alpha)} = \sup_{\substack{\alpha \in \Acal \\ V^\alpha_0 = v_0}} \mathbb{E}^\ast\left[u(V_t^\alpha)\right],
\end{equation}
where the expectation on the right-hand side is taken with respect to the new probability measure $\PW^\ast$. Similarly to the previous proofs, the problem above can be dealt with from the perspective of an infinite dimensional maximization problem. We then rewrite \eqref{eq:ch_of_measure} as
\begin{equation*}
    \sup \left\{\mathbb{E}^\ast\left[u(\xi)\right]\mid \xi \in L^1(\Omega,\Fcal_t,\PWq) \text{ and } \Eq{\xi} = v_0\right\}.
\end{equation*}  
Observe that, whenever the function $u:\R \to \R \cup \{-\infty\}$ is of exponential form, as in Eq. \eqref{eq:exp_ut_deterministic}, the above problem is analogous to the one of Proposition \ref{prop:det_exp_ut}. In particular, defining the $\PW^\ast$-martingale $\varphi_t := \frac{Z_t}{X_t}$ which has dynamics $d\varphi_t = \varphi_t (\theta_t - \beta_t)dW^\ast_t$, we can directly apply the results of Proposition \ref{prop:det_exp_ut} to obtain that the optimization problem is consistent if and only if $\left(\theta_t - \beta_t\right)^2 = k(t)$, where $k:\R \to \R$ is some deterministic function. 

Lastly, assuming $u(x) = -\frac{1}{\gamma} e^{-\gamma x}$, we show that $\mbu_t(\om, x) = u(x) \cdot X_t(\om)$ with initial datum $u_0(x) = -\frac{1}{\gamma} e^{-\gamma x}$ is a forward performance if and only if $k(t) = 0$. 
Observe that, given the dynamics of $V_t^\alpha$ in $\eqref{discounted:port}$ we have that 
$$
d\left(u(V_t^\alpha)\right) = V^\alpha_t e^{-\gamma V_t^\alpha} \left[\left((\mu_t - r) \alpha_t - \frac12 \gamma V^\alpha_t \sigma^2_t \alpha_t^2\right)dt + \sigma_t \alpha_t dW_t\right]
$$
and hence the stochastic differential of $\mbu_t(V_t^\alpha)  = u(V_t^\alpha) \cdot X_t$ is given by
\begin{equation}\label{eq:parabola_multiplicative}
    d\mbu_t(V_t^\alpha) = X_t e^{-\gamma V_t^\alpha}\left[\left(-\frac12 \gamma (V_t^\alpha \sigma_t \alpha_t)^2 - V_t^\alpha \sigma_t \alpha_t(\theta_t - \beta_t)\right)dt + \left(V_t \sigma_t \alpha_t - \frac{1}{\gamma} \beta_t \right)dW_t \right].
\end{equation}
Notice that the drift coefficient of the latter describes a downward parabola. In particular, requiring that the drift be non-positive for all $\alpha \in \Acal$ is equivalent to imposing that the $y$-coordinate of the vertex of the parabola be 0. This latter condition then implies that $(\theta_t - \beta_t)^2 = k(t) = 0$. 

Finally, for this choice of $k(t)$ a simple inspection of \eqref{eq:parabola_multiplicative} shows that the unique strategy $\alpha^\ast \in \Acal$ such that $\left(\mbu_t(V_t^{\alpha^\ast})\right)_{t\geq 0}$ is a $\PW$-martingale is $\alpha^\ast_t = 0$. 

\section{Additional material}

\begin{lemma}\label{utile} Let $\mbu:\Omega\times \R\to \R\cup \{-\infty\}$, be $\Fcal_t\times \Borel_{\R}$-measurable, with $\mbu(\omega,\cdot)$ strictly increasing and strictly concave for any $\omega\in\Omega$. If for any $\omega\in\Omega$, $\mbu(\omega,\cdot)$ is differentiable and its derivative is invertible, then for $\PWq\sim \PW$ the problem  
    \begin{equation*}
        \sup\left\{\int_{\Omega}\mbu(\omega,\xi(\omega))d\PW\mid \xi \in L^1(\Omega,\Fcal_t,\PWq) \text{ and } \Eq{\xi} = v_0 \right\}.
    \end{equation*}  
has a unique solution $\xi^{\ast}(\omega)= F(\omega, \lambda^{\ast} Y(\omega))$ where $Y=\frac{d\PWq}{d\PW}$, $F(\omega,y)=(\mbu')^{-1}(\omega,y)$, $\lambda^{\ast}$ solves the equation $\Eq{\xi^{\ast}} = v_0$ (provided that 
\(\Eq{|\xi^*|}<\infty, \Ep{|\mbu(\xi^*)|}<\infty
\)). 
\end{lemma}
\begin{proof}
The Lagrangian function associated to the optimization problem reads
    $$
    \int_\Om \mbu(\xi) d\PW - \lambda\left(\int_\Om \xi d\PWq - v_0\right).
    $$
  and by a can change of probability measure we obtain 
    $$
    \int_\Om \left(\mbu(\xi) - \lambda\cdot Y(\xi - v_0)\right)d\PW.
    $$
    Indeed the pointwise first order condition with respect to $\xi$ becomes 
    $$
    \mbu'(\omega,\xi) - \lambda Y(\omega) = 0,
    $$
    so that the candidate for the optimum is $F(\omega, \lambda Y(\omega))$.
    Simple inspections show that for any $\xi \in L^1(\PWq)$ with  $\Eq{\xi} = v_0$ we have  
    $$\int_\Om \mbu(\xi)d\PW \leq \int_\Om \mbu(F(\cdot, \lambda Y)) d\PW - \lambda\left(\int_\Om F(\cdot, \lambda Y) d\PWq - v_0\right)$$
    and therefore if $\lambda^{\ast}$ solves the equation $\Eq{F(\cdot, \lambda^{\ast} Y)} = v_0$ then $F(\omega, \lambda^{\ast} Y(\omega))$ is optimal. Uniqueness follows from strict monotonicity and concavity.
\end{proof}

 \subsection{Proof of Remark \ref{power:utility}}\label{proof:power}
Fix $t>0$. Analogously to the previous proof we reformulate Problem \eqref{optimization} as its infinite dimensional counterpart as in Eq. \eqref{static}. Considering a power utility function, the optimal profile $\xi^\ast_t$ takes the form
\begin{equation*}
    \xi^\ast_t = \frac{v_0 \cdot Z_t^\beta}{\Eq{Z_t^\beta}},
\end{equation*}
where $\beta = -\frac{1}{1-\gamma}$ and $Z_t = \frac{d\PWq}{d\PW}$. Following a standard argument (see e.g. \cite{Bjork09}) let us define 
\begin{equation*}\label{eq:zeta_zero}
    Z_t^0 := \exp\left\{-\frac12 \int_0^t \theta_u^2 (\beta +1)^2 du + \int_0^t \theta_u (\beta + 1) dW_u\right\}
\end{equation*}
and observe that $Z^{\beta+1}_t = Z_t^0 \exp\left\{\frac12\int_0^t \theta^2_u \beta^2 \gamma du \right\}$. Denote by $\QW_0$ the probability measure induced by $Z_t^0 = \frac{d\QW_0}{d\PW}$ and by 
\begin{equation*}\label{eq:H0}
    H_t = \Eq{Z_t^\beta} = \Ezero{\exp\left\{\frac12\int_0^t \theta^2_u \beta^2 \gamma du \right\}}.
\end{equation*}
For some time $s < t$ let $\Pi_s$ denote the discounted no-arbitrage price of the optimum $\xi^\ast_t$. The problem is thus time-consistent whenever $\Pi_s = \xi_s^\ast$ that is, whenever
\begin{equation}\label{eq:consistency_cond_powut}
    \frac{1}{H_t}\cdot\Eq{Z_t^{\beta}\middle| \Fcal_s}  = \frac{ Z_s^{\beta}}{H_s}.
\end{equation}
On the left hand side of \eqref{eq:consistency_cond_powut} we have that 
\begin{align*}
    \Eq{Z_t^{\beta}\middle| \Fcal_s} &= \frac{\Ep{Z_t^{\beta+1}\middle| \Fcal_s}}{\Ep{Z_t\middle| \Fcal_s}}\nonumber = \frac{\Ep{Z_t^{\beta+1}\middle| \Fcal_s}}{Z_s},\nonumber 
\end{align*}
Moreover simple computations leads to
\begin{align*}
    \Ep{Z_t^{\beta+1}\middle| \Fcal_s} 
                     &= \underbrace{Z_s^0 \exp \left\{\frac12\int_0^s\theta_u^2 \beta^2 \gamma du \right\}}_{Z_s^{\beta + 1}} \Ezero{\exp \left\{\frac12\int_s^t\theta_u^2 \beta^2 \gamma du \right\}\middle|\Fcal_s},
\end{align*}

and therefore \eqref{eq:consistency_cond_powut} becomes
\begin{equation}\label{eq:ht_hs}
    \frac{H_t}{H_s} = \Ezero{\exp \left\{\frac12\int_s^t\theta_u^2 \beta^2 \gamma du \right\}\middle|\Fcal_s}.
\end{equation}
We define the following process
\begin{equation*}
    M_t := \frac{\exp\left\{\frac12\int_0^t\theta_u^2 \beta^2 \gamma du\right\}}{\Ezero{\exp \left\{\frac12\int_0^t\theta_u^2 \beta^2 \gamma du \right\}}}
\end{equation*}
and observe that 
\begin{align}\nonumber
    \Ezero{M_t\middle|\Fcal_s} &= \frac{\Ezero{\exp\left\{\frac12\int_0^t\theta_u^2 \beta^2 \gamma du\right\}\middle|\Fcal_s}}{\Ezero{\exp \left\{\frac12\int_0^t\theta_u^2 \beta^2 \gamma du \right\}}} \\
    &= \frac{\Ezero{\exp\left\{\frac12\int_s^t\theta_u^2 \beta^2 \gamma du\right\}\middle|\Fcal_s}\exp\left\{\frac12\int_0^s\theta_u^2 \beta^2 \gamma du\right\}}{\Ezero{\exp \left\{\frac12\int_0^s\theta_u^2 \beta^2 \gamma du \right\}\Ezero{\exp \left\{\frac12\int_s^t\theta_u^2 \beta^2 \gamma du \right\}\middle|\Fcal_s}}}.\nonumber
\end{align}
Assuming Eq. \eqref{eq:ht_hs} holds we have that $\Ezero{\exp \left\{\frac12\int_s^t\theta_u^2 \beta^2 \gamma du \right\}\middle|\Fcal_s} \in \R$ and hence we can take it out of the expectation to obtain
\begin{equation*}
    \Ezero{M_t\middle|\Fcal_s} = \frac{\exp\left\{\frac12\int_0^s\theta_u^2 \beta^2 \gamma du\right\}}{\Ezero{\exp \left\{\frac12\int_0^s\theta_u^2 \beta^2 \gamma du \right\}}} = M_s
\end{equation*}
and we conclude $M_t$ is a $\PWq_0$-martingale. A simple inspection show that $M_t$ is of finite variation and by Proposition IV.1.2 \cite{RY99} it follows that it must be constant. In particular then $M_t = M_0 = 1$ which entails  
\begin{equation*}\label{eq:exp_deterministic}
\exp\left\{\frac12\int_0^t\theta_u^2 \beta^2 \gamma du\right\} = \Ezero{\exp\left\{\frac12\int_0^t\theta_u^2 \beta^2 \gamma du\right\}},
\end{equation*}
for any $t \in [0,\infty)$. Observing that we can differentiate with respect of time $t$ under the expectation operator, taking the derivative on both sides we find that  $\Ezero{\theta_t^2} = \theta_t^2$
and since the above holds for any $t \in [0,\infty)$, we conclude.  \hfill \qed

\section*{Disclosure of interest}
No potential competing interest is reported by the authors.


\begin{thebibliography}{10}

\bibitem{BC10}
S.~Basak and G.~Chabakauri.
\newblock Dynamic mean-variance asset allocation.
\newblock {\em Rev. Financ. Stud.}, 23(8):2970--3016, 2010.

\bibitem{BDM24}
Berton, E., Doldi, A. and Maggis, M., On {C}onditioning and {C}onsistency for {N}onlinear {F}unctionals. {\itshape SIAM J. Financial Math.}, 2025, \textbf{16}, 667--691.

\bibitem{Bjork09}
T.~Bj\"ork.
\newblock {\em Arbitrage Theory in Continuous Time}.
\newblock Oxford Finance Series. Oxford University Press, 2009.

\bibitem{BKM17}
T.~Bj\"ork, M.~Khapko, and A.~Murgoci.
\newblock On time-inconsistent stochastic control in continuous time.
\newblock {\em Finance Stoch.}, 21(2):331--360, 2017.

\bibitem{BMZ14}
T.~Bj\"ork, A.~Murgoci, and X.~Y. Zhou.
\newblock Mean-variance portfolio optimization with state-dependent risk
  aversion.
\newblock {\em Math. Finance}, 24(1):1--24, 2014.

\bibitem{EL06}
I.~Ekeland and A.~Lazrak.
\newblock The golden rule when preferences are time inconsistent.
\newblock {\em Math. Financ. Econ.}, 4(1):29--55, 2010.

\bibitem{Gilboa09}
I.~Gilboa.
\newblock {\em Theory of decision under uncertainty}, volume~45 of {\em
  Econometric Society Monographs}.
\newblock Cambridge University Press, Cambridge, 2009.

\bibitem{HH07}
V.~Henderson and D.~Hobson.
\newblock Horizon-unbiased utility functions.
\newblock {\em Stochastic Process. Appl.}, 117(11):1621--1641, 2007.

\bibitem{JZ08}
H.~Jin and X.~Y. Zhou.
\newblock Behavioral portfolio selection in continuous time.
\newblock {\em Math. Finance}, 18(3):385--426, 2008.

\bibitem{KP77}
F.~E. Kydland and E.~C. Prescott.
\newblock Rules rather than discretion: The inconsistency of optimal plans.
\newblock {\em J. Pol. Econ.}, 85(3):473--491, 1977.

\bibitem{LN00}
D.~Li and W.-L. Ng.
\newblock Optimal dynamic portfolio selection: multiperiod mean-variance
  formulation.
\newblock {\em Math. Finance}, 10(3):387--406, 2000.

\bibitem{Markowitz53}
H.~Markowitz.
\newblock Portfolio selection.
\newblock {\em J. Finance}, 7(1):77--91, 1952.

\bibitem{Merton69}
R.~C. Merton.
\newblock Lifetime portfolio selection under uncertainty: The continuous-time
  case.
\newblock {\em Rev. Econom. Statist.}, 51(3):247--257, 1969.

\bibitem{MZ07}
M.~Musiela and T.~Zariphopoulou.
\newblock Investment and valuation under backward and forward dynamic
  exponential utilities in a stochastic factor model.
\newblock In {\em Advances in mathematical finance}, Appl. Numer. Harmon.
  Anal., pages 303--334. Birkh\"auser Boston, Boston, MA, 2007.

\bibitem{MZ08}
M.~Musiela and T.~Zariphopoulou.
\newblock Optimal asset allocation under forward exponential performance
  criteria.
\newblock In {\em Markov processes and related topics: a {F}estschrift for
  {T}homas {G}. {K}urtz}, volume~4 of {\em Inst. Math. Stat. (IMS) Collect.},
  pages 285--300. Inst. Math. Statist., Beachwood, OH, 2008.

\bibitem{MZ09}
M.~Musiela and T.~Zariphopoulou.
\newblock Portfolio choice under dynamic investment performance criteria.
\newblock {\em Quant. Finance}, 9(2):161--170, 2009.

\bibitem{RY99}
D.~Revuz and M.~Yor.
\newblock {\em Continuous martingales and {B}rownian motion}, volume 293 of
  {\em Grundlehren der mathematischen Wissenschaften [Fundamental Principles of
  Mathematical Sciences]}.
\newblock Springer-Verlag, Berlin, third edition, 1999.

\bibitem{Samuelson69}
P.~A. Samuelson.
\newblock Lifetime portfolio selection by dynamic stochastic programming.
\newblock {\em Rev. Econom. Statist.}, 51(3):239--246, 1969.

\bibitem{WZ99}
P.~P. Wakker and H.~Zank.
\newblock State dependent expected utility for {S}avage's state space.
\newblock {\em Math. Oper. Res.}, 24(1):8--34, 1999.

\bibitem{ZL00}
X.~Y. Zhou and D.~Li.
\newblock Continuous-time mean-variance portfolio selection: a stochastic {LQ}
  framework.
\newblock {\em Appl. Math. Optim.}, 42(1):19--33, 2000.

\bibitem{Zitkovic09}
G.~Žitković.
\newblock A dual characterization of self-generation and exponential forward
  performances.
\newblock {\em Annals Appl. Prob.}, 19(6):2176--2210, 2009.

\end{thebibliography}
\end{document}